\documentclass{sig-alternate}

\usepackage{algorithm}
\usepackage[noend]{algpseudocode}
\usepackage{graphicx}
\usepackage{tabularx}
\usepackage{amsmath}

\newtheorem{theorem}{Theorem}
\newtheorem{lemma}{Lemma}
\newcommand{\Geqt}{\ensuremath{G_\text{eqt}}}

\begin{document}
\conferenceinfo{NANOCOM}{2015 Boston, Massachusetts, USA}
\title{An Algorithmic Framework for Shape Formation Problems in Self-Organizing Particle Systems
}

\numberofauthors{5} 
\author{
\alignauthor
Zahra Derakhshandeh
\\
       {\small\affaddr{Computer Science and Engineering, CIDSE,}}\\
       {\small\affaddr{Arizona State University, USA}}\\
       {\tiny\email{zderakhs@asu.edu}}
\alignauthor
Robert Gmyr
\\
       {\small\affaddr{Dept. of Computer Science,}}\\
       {\small\affaddr{University of Paderborn, Germany}}\\
       {\tiny\email{gmyr@mail.upb.de}}
\alignauthor Andr\'ea W.\ Richa
\\
       {\small\affaddr{Computer Science and Engineering, CIDSE,}}\\
       {\small\affaddr{Arizona State University, USA}}\\
       {\tiny\email{aricha@asu.edu}}
\and  
\alignauthor Christian Scheideler\\
       {\small\affaddr{Dept. of Computer Science,}}\\
       {\small\affaddr{University of Paderborn, Germany}}\\
       {\tiny\email{scheideler@upb.de}}
\alignauthor Thim Strothmann\\
       {\small\affaddr{Dept. of Computer Science,}}\\
       {\small\affaddr{University of Paderborn, Germany}}\\
       {\tiny\email{thim@mail.upb.de}}
}

\maketitle
\begin{abstract}
Many proposals have already been
  made for realizing programmable matter, ranging from shape-changing molecules,
  DNA tiles, and synthetic cells to reconfigurable modular
  robotics. Envisioning systems of nano-sensors devices, we are particularly interested in programmable matter consisting of
  systems of simple computational elements, called {\em particles}, that can establish and release
  bonds and can actively move in a self-organized way,
  and in shape formation problems relevant for programmable matter in those self-organizing particle systems (SOPS). In this paper, we present a general algorithmic framework for shape formation problems in  SOPS, and show direct applications of this framework to the problems of having the particle system self-organize to  form a hexagonal or triangular shape. Our algorithms utilize only local control, require only constant-size memory particles, and are asymptotically optimal both in terms 
     of the total number of movements needed to reach the desired shape configuration. 
  \end{abstract}

\section{Introduction}
Imagine that we had a piece of matter that can change its physical
properties
like shape, density, conductivity, or color in a programmable fashion
based on
either user input or autonomous sensing. This is the vision behind what is
commonly known as {\em programmable matter}.
Programmable matter has been the subject of many recent novel distributed
computing proposals, ranging from shape-changing molecules,
  DNA tiles, and synthetic cells to reconfigurable modular
  robotics. Each of these proposals pursued solutions for  specific
application
scenarios with their own, special capabilities and constraints.

We envision systems of nano-sensors devices that will have very limited
computational capabilities individually, but which can 
collaborate to reach a lot more as a collective. Ideally, those
nano-sensor devices will be able to self-organize in order to achieve a
desired collective goal without the need of central control or external
(in particular, human) intervention. For example, one could envision using
a system of self-organizing nano-sensor devices  to
identify and coat (and possibly repair) leaks on a nuclear reactor without
the need for human intervention; self-organizing systems of nano-sensor
devices could also be used to monitor environmental and structural
conditions in abandoned mines, on the exterior of an airplane or
spacecraft, bridges and other structures, possibly also
self-repairing the structure--- i.e., realizing what
has been coined as "smart paint". The applications in the health arena are
also endless, e.g., self-organizing nano-sensor devices could be used
within our bodies to detect and coat an area where internal bleeding
occurs, eliminating the need of immediate surgery, or they could be used
to identify and isolate tumor/malignouos cells. In many applications,
there may be a specific shape that one would like the system to assume
(e.g., a disc, or a line, or even any compact shape).

Hence, from an algorithmic point-of-view, we are interested in
programmable matter consisting of
  systems of simple computational elements, called {\em particles}, that
can establish and release (communication or physical)
  bonds and can actively move in a self-organized way,
  and in general shape formation problems 
  in those self-organizing particle systems (SOPS).

\vspace{-.2in}
\subsection{Geometric Amoebot model} \label{sec:model}

We will use the geometric amoebot model presented in~\cite{arxivDNA, spaa-ba14} as our basic model for SOPS.

In all of our shape formation algorithms, the 
set of particles will maintain 
a connected structure at all times. 
We assume that we have a graph $G(V,E)$ that represents the relative positions that a connected set of particles may assume --- i.e., 
$V$ represents all possible positions of a particle (relative to the other
particles in their structure) and $E$ represents all possible transitions
between nodes. In the geometric amoebot model we assume that $G=$ $\Geqt$, where $\Geqt$ is the infinite regular  triangular grid graph\footnote{The triangular grid graph $\Geqt$ is the dual graph of a regular hexagonal tiling in 2D space.}(see Part (a) of Figure~\ref{fig:graph}).

\begin{figure*}
    \centering
    \includegraphics[scale=0.4]{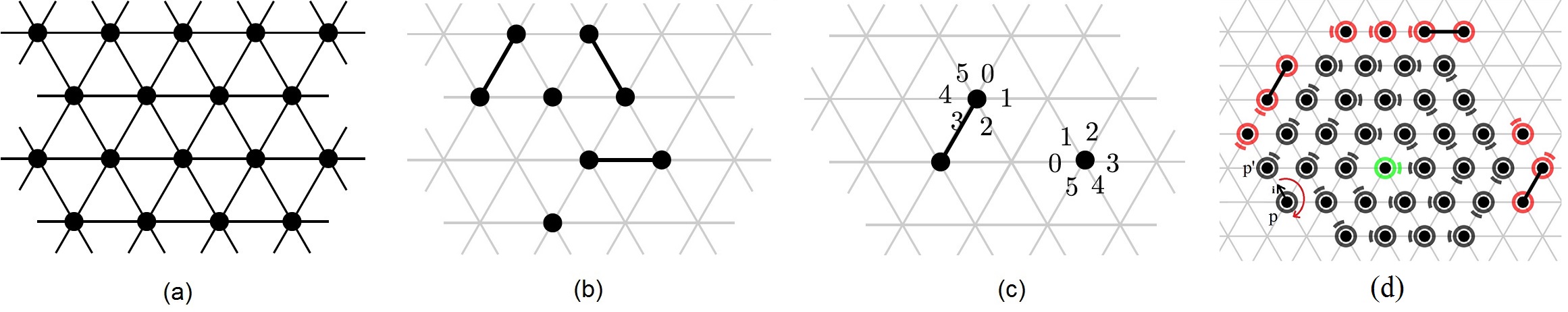}
		\vspace{-.15in}
    \caption{\small
			Part (a) shows a section of $\Geqt$; 
    	nodes of $\Geqt$ are shown as black circles.
    	Part (b) shows five particles on $\Geqt$: 
			 the underlying graph  $\Geqt$ as a gray mesh; a particle occupying a single node is depicted as a black circle,
    	and a particle occupying two nodes is depicted as two black circles connected by an edge.
			Part (c) depicts two particles occupying two non-adjacent positions on $\Geqt$; the particles have different offsets for their head bond labelings. 
			Part (d) shows an intermediate configuration of the HEX algorithm. The seed is depicted in green, retired particles are black,  and roots are red. Particle $p$ is the last added particle to the retired structure. Hence, edge $i$ connects $p$ to the retired particle $p'$ (edge $i$ has the flag $p'.snakedir$). The red arrow depicts the process of setting $p.snakedir$ in clockwise manner for $p$. 
    }
    \label{fig:graph}
\end{figure*}

We briefly recall the main properties of the geometric amoebot model. Each particle occupies either a single node or a pair of adjacent nodes in
$\Geqt$, and every node can be occupied by at most one particle. Two particles
occupying adjacent nodes are \emph{connected}, and we refer to such particles
as \emph{neighbors}. 

Particles move through \emph{expansions} and \emph{contractions}: If a
particle occupies one node (i.e., it is \emph{contracted}), it can expand to
an unoccupied adjacent node to occupy two nodes. If a particle occupies two
nodes (i.e., it is \emph{expanded}), it can contract to one of these nodes to
occupy only a single node. Performing movements via expansions and
contractions may represent the way particles physically move, or may be seen as a logical "look-ahead and then move" logical operation. It has several advantages, including allowing particles to abort a movement if there is a conflict (see~\cite{arxivDNA} for more details). 
A particle always knows whether
it is contracted or expanded --- in the latter, it also knows along which edge it expands --- and
this information will be available to neighboring particles.
A \emph{handover} allows particles to stay connected as they move; two scenarios are possible: a) a contracted particle $p$
can "push" a neighboring expanded particle $q$ and expand into the neighboring
node previously occupied by $q$, forcing $q$ to contract, or b) an expanded
particle $p$ can "pull" a neighboring contracted particle $q$ to a cell
occupied by it thereby expanding that particle to that cell, which allows $p$
to contract to its other cell.
In part(b) of Figure~\ref{fig:graph}, we illustrate a set of particles (some
contracted, some expanded) on the underlying graph $\Geqt$.

Particles are \emph{anonymous} but the bonds of each
particle have unique labels, which implies that       
a particle can uniquely identify each of its outgoing edges.
Moreover, for each particle the bonds are labeled in a consecutive way in clockwise
direction so that every particle has the same sense of clockwise direction,
but the particles may not have a common sense of orientation in a sense that
they have different offsets of the labelings (see Figure~\ref{fig:graph}, Part (c)).
Each particle has a constant-size local memory in which it can store some bounded amount of
information, and any pair of connected particles has a bounded shared memory
that can be read and written by both of them and that can be accessed using
the edge label associated with that connection.
 We assume the standard asynchronous model from distributed computing,
where the system of particles progresses by performing atomic actions, each of which affects the configuration of one or two particles. 
Whenever a particle is activated (i.e., performs an atomic action), it can perform an
arbitrary bounded amount of computation (involving its local memory as well as
the shared memories with its neighboring particles) followed by no or a single
movement. A \emph{round} is over once every particle has been activated at
least once.

\subsection{Our Contributions}

In this paper, we present a {\em general algorithmic  framework for shape formation problems} in SOPS, which constitutes of two basic algorithmic primitives:  the {\em spanning forest} primitive and the {\em snake formation} primitive. We present concrete applications of these two primitives to two specific shape formation problems, namely to the problems of  having the system of particles self-organize to form a {\em hexagonal shape} and to form a {\em triangular shape}. Both the hexagonal shape and the triangular shape formation algorithms are  {\em optimal} with respect to {\em work}, which we measure by the total number of particle movements needed to reach the desired shape configuration, 
as we prove in Theorems~\ref{thm:hexagon} and \ref{thm:triangle}. Our algorithms rely only on local information (e.g., particles do not have ids, nor do they know $n$, the total number of particles, or have any sort of global coordinate/orientation system), and require only constant-size memory particles. 

\subsection{Related Work}

Many approaches related to programmable matter have recently been proposed.
One can distinguish between active and passive systems. In passive systems (e.g., DNA computing~\cite{Adl94,BDLS96,
WLWS98}, tile self-assembly systems~\cite{doty2012,patitz2014,Woods2013intrinsic}),) the
particles either do not have any intelligence at all (but just move and bond
based on their structural properties or due to chemical interactions with the
environment), or they have limited computational capabilities but cannot
control their movements. 
We will not describe passive models in detail as they are only of little relevance for our approach.
On the other hand in \emph{active systems}, computational
particles can control the way they act and move in order to solve a specific task.
Robotic swarms, and modular robotic systems are
some examples of active programmable matter systems.

In the area of \textit{swarm robotics} it is usually assumed that there is a collection of autonomous robots
that have limited sensing, and communication ranges, and that can freely move in a given area.
They follow a variety of goals, 
 including for example shape formation problems (e.g.,~\cite{fl08,kilobots}). 
Surveys of recent results in swarm robotics can be found in~\cite{Ker12,McL08}.
While the analytical techniques developed in the area of swarm robotics and natural swarms are of some relevance for this work, the individual units in those systems have more powerful communication and processing capabilities than in the systems we consider.

The field of \textit{modular self-reconfigurable robotic systems} focuses on intra-robotic aspects
such as the design, fabrication, motion planning, and control of autonomous kinematic machines with variable morphology
(see e.g.,~\cite{FNKB88,YSS+07}).
\textit{Metamorphic robots}  form a subclass of self-reconfigurable robots that share some of the characteristics of our geometric model~\cite{Chi94}.
The hardware development in the field of self-reconfigurable robotics has been complemented
by a number of algorithmic advances (e.g.,~\cite{BKRT04,WWA04,kilobots}),
but so far
mechanisms that automatically scale from a few to hundreds or thousands of individual units are still under investigation,
and no rigorous theoretical foundation is available yet.

The \emph{nubot} model~\cite{winfree13,chen2013parallel} 
aims at providing the theoretical
framework that would allow for
a more rigorous algorithmic study of biomolecular-inspired systems,
more specifically of self-assembly systems with active molecular
components. While bio-molecular inspired systems share many similarities with our SOPS, there are many differences --- e.g.,
 there is always an arbitrarily large supply of "extra" particles that can be added to the system as needed, and the system allows for an additional (non-local) notion of rigid-body movement.

\section{Shape Formation}
\label{sec:shapeFormation}

In this paper we focus on solving \emph{shape formation problems} in the geometric amoebot model 
 starting from any initial connected configuration of particles.
 We present a general algorithmic framework for shape formation problems and then specifically  we investigate the \emph{Hexagonal Shape Formation (HEX)}  and the \emph{Triangular Shape Formation (TRI)}  problems where the desired shape is a hexagon and a triangle respectively.  We formally define a shape formation problem as a tuple $\mathcal{M = (I, G)}$
where $\mathcal{I}$ and $\mathcal{G}$ are sets of connected configurations.
We say $\mathcal{I}$ is the set of possible initial configurations and $\mathcal{G}$ is the set of goal configurations.
 Accordingly, for the HEX problem, $\mathcal{G}$ would be all configurations where the positions of the set of particles induce a hexagon on $\Geqt$ (note that depending on the number of particles the constructed shape may not necessarily be a perfect hexagon since the outer layer of the hexagon may not be fully complete). Similarly, for the TRI problem, $\mathcal{G}$ is equal to  the set of all configurations that constitute a triangle in $\Geqt$ (except for possibly the outer layer of the triangle, which may be partially full).
 We say an algorithm $\mathcal{A}$ 
\emph{solves} a shape formation problem $\mathcal{M}$ if for any execution of $\mathcal{A}$ on a system in an arbitrary configuration from $\mathcal{I}$,
$\mathcal{A}$ \emph{terminates} (i.e., the execution eventually reaches a configuration in which each particle does not move anymore) in one of the valid configurations in $\mathcal{G}$.

Before we proceed, we provide some preliminaries. 
For all algorithms we assume that there is a specific particle we call the \emph{seed} particle, which provides the starting point for constructing the respective shape. If a seed is not available, one can be chosen using the leader election algorithm proposed in \cite{arxivDNA} . We define the set of {\em states} that a particle can be in as
\emph{inactive}, \emph{follower}, \emph{root}, and \emph{retired}. Initially, all particles are inactive, except the seed particle, which is always in a {\em retired} state.
In addition to its state, each particle $p$ may maintain a constant number of {\em flags}
in its shared memory.
For an expanded particle, we denote the node the particle last expanded into
as the \emph{head} of the particle and call the other occupied node its
\emph{tail}.
In our algorithm, we assume that every time a particle contracts, it contracts out of its tail. Note that with this convention, the node occupied by the head of a particle still is occupied by that particle after a contraction. 
Part (c) of Figure~\ref{fig:graph} shows an example of the labeling of the heads of two particles on $\Geqt$.

Generally speaking, the shape formation algorithms we propose for hexagonal and triangular shapes progress as follows. Particles organize themselves into a {\em spanning set of disjoint trees} where the roots of the trees are non-retired particles adjacent to the partially constructed shape structure (consisting of all retired particles). Root particles lead the way by moving in a predefined direction around the current structure. 
 
The remaining particles (i.e., the followers) follow behind the leading root particles, hence the system flattens out towards the direction of movement. 
Once the leading particles reach a valid position where the shape can be extended (following the rules for the {\em snake formation} for the particular shape), they stop moving and change their state to retired as well. 
This process continues until all particles become retired.  Note that the spanning forest component of this general approach is the same for any shape formation algorithm: It is only in the rules that determine the next valid position to be filled in the shape structure being built that the respective algorithms differ. We determine the next valid position to be filled sequentially following a {\em snake} (i.e., a line of consecutive positions in $\Geqt$),  that fills in the space of the respective shape structure and scales naturally with the number of particles in the system.

\subsection{Spanning Forest Algorithm}
\label{sec:spanningForestAlgorithm}
The \emph{Spanning Forest} algorithm primitive, given in Algorithm~\ref{alg:spanningForestAlgorithm},
is a building block we use for all of our shape formation problems. This primitive was also used in~\cite{arxivDNA}, where we present a preliminary self-organizing algorithm for forming a straight line of particles. We present the algorithm here for completeness.
 Each particle $p$ continuously runs
Algorithm~\ref{alg:spanningForestAlgorithm} until it becomes retired. If
particle $p$ is a follower, it stores a flag $p.parent$ in its shared memory
corresponding to the edge adjacent to its parent $p'$ in the spanning forest
(any particle $q$ can then easily check if $p$ is a child of $q$).

	Initially all system particles, except the seed, are inactive. In a nutshell, the particles that are touching the seed or other retired particle 
 become roots; the root particles move around the partially constructed shape structure in a clockwise manner until they find a valid position on the snake and become retired; follower particles follow the movement of the respective root until they become roots themselves.
As we will see later, the initial snapshots of Figures~\ref{fig:hexagonsnapshots} and~\ref{fig:trianglesnapshots} illustrate the spanning forest formation for the respective initial particle configurations.

\begin{algorithm*}[htb]
    Depending on $p$'s current state, a particle $p$ behaves as described below:
		
        \begin{tabularx}{\textwidth}{lX}
        \textbf{inactive}: &
                 If $p$ is connected to a retired particle, then $p$ becomes a {\em root} particle. Otherwise,
             if an adjacent particle $p'$ is a root or a follower,
        $p$ sets the flag $p.parent$ on the shared memory corresponding to the edge to $p'$ and becomes a {\em follower}. If none of the above applies, it remains inactive.
        \\
        \textbf{follower}: &
        If $p$ is contracted and connected a retired particle, then $p$ becomes a {\em root} particle.
         Otherwise, it considers the following three cases: $(i)$ if $p$ is contracted and $p$'s parent $p'$  is expanded, then $p$ expands in the direction given by $p.parent$ in a handover with $p'$, and may need to adjust $p.parent$ to still point to particle $p'$ after the handover;
         $(ii)$ if $p$ is expanded and has a contracted child particle $p'$,
         then $p$ executes a handover with $p'$; $(iii)$ if $p$ is expanded,  has no children,
        and $p$ has no inactive neighbor, then $p$ contracts.
        \\
        \textbf{root}: &
       Particle $p$ runs the corresponding snake formation algorithm (Algorithm~\ref{alg:retiredConditionHexagon} or~\ref{alg:retiredConditionTriangle}, for HEX or TRI resp.), and becomes {\em retired} accordingly. Otherwise, it considers the following three cases: $(i)$ if $p$ is contracted, it tries to expand in the direction given by {\sc RootDirection} $(p)$; $(ii)$ If $p$ is expanded and has a child $p'$, then
            $p$ executes a handover contraction with $p'$;
         $(iii)$ if $p$ is expanded and has no children,
        and no inactive neighbor, then $p$ contracts.
        \\
        \textbf{retired}: &
            $p$ performs no further action.
        \\
    \end{tabularx}
		
		\vspace{.05in}
		
	{\sc RootDirection} $(p)$: \\
    \vspace{-.2in}
   \begin{algorithmic}
				\State Let $i$ be the label of an edge connected to a retired particle.
    \While{edge $i$ points to a retired particle}
        \State $i \; \gets \;$ label of next edge in clockwise direction
     \EndWhile
    \State \textbf{return} $i$
		\end{algorithmic}
		
    \caption{Spanning Forest Algorithm for Shape Formation}
    \label{alg:spanningForestAlgorithm}
\end{algorithm*}

\subsection{Hexagonal Shape Formation}
\label{sec:HexagonFormation}
We now investigate the \emph{Hexagonal Shape Formation (HEX)} problem where the system of particles has to assume the shape of a hexagon (but for the outer layer, which may not be completely full) in $\Geqt$.
The hexagon will be constructed around the seed particle. Note that a hexagon in $\Geqt$ is actually a disk, since it can be defined by the set of all nodes of $\Geqt$ within a certain distance $r$ from a seed node.
 
We will organize the particles according to a spiral snake structure which will incrementally add new layers to the hexagon, scaling naturally with the number of particles in the system. 
 In order to characterize the snake formation for a given shape formation problem, one only needs to specify the direction in which the line of particles forming the snake should continue to grow, for each new particle added to the snake. Hence once a particle finds the next valid position on the snake, it will become retired and set the snake direction accordingly (by correctly setting the flag $p.snakedir$ on the respective edge). 
 Different rules for snake formation will realize different shapes. In particular,  Algorithm~\ref{alg:retiredConditionHexagon} specifies the rules for the spiral snake formation for HEX. 

Initially, the seed particle $p$ sets the flag $p.snakedir$ in the shared memory 
corresponding to one of its adjacent edges (e.g., the edge with label 0). 
 Any particle adjacent to a retired particle becomes a root following the spanning forest algorithm. 
 Each root $p$ moves in a clockwise fashion around the structure of retired particles 
 until it finds the next position to extend the hexagonal snake (i.e., a position connected to a retired particle via an edge flagged $p.snakedir$) and becomes retired, following Algorithm~\ref{alg:retiredConditionHexagon} 
(see Part (d) of Figure~\ref{fig:graph}). 

\begin{algorithm}[htb]
	\begin{algorithmic}
	  \If {$p$ is a contracted root}
    \If {$p$ has an adjacent edge $i$ to $p'$ with a flag $p'.snakedir$, where $p'$ is retired} \Comment{retired condition}
		\While {edge $i$ is connected to a retired particle}
			\State $i \; \gets \;$ label of next edge in clockwise direction	  
		\EndWhile
    \State $p$ sets the flag $p.snakedir$ for edge $i$  
    \State $p$ becomes {\em retired}.
    \EndIf
    \EndIf
  \end{algorithmic}
	\caption{Snake Formation for HEX}
	\label{alg:retiredConditionHexagon}
\end{algorithm}

Figure~\ref{fig:hexagonsnapshots} depicts some snapshots of a run of HEX algorithm. See Appendix for the proof of the following theorem. 

	\vspace{-.1in}
\begin{theorem}
	\label{thm:hexagon}
	Our algorithm solves the HEX problem in worst-case optimal $O(n^2)$ work. 
\end{theorem}

\subsection{Triangular Shape Formation}
\label{sec:TriangleFormation}
Now we investigate the \emph{Triangular Shape Formation} problem (TRI) where the system of particles has to assume a final triangular shape on $\Geqt$ (but for possibly the outer layer). 

As we discussed for the HEX problem, in order to solve the TRI problem in our Spanning Forest + Snake Formation algorithmic framework, one only needs to setup the correct rules for growing a "triangular snake", which will be accomplished by Algorithm~\ref{alg:retiredConditionTriangle}. The snake formation rules for the TRI problem are complex than the ones we had for the HEX problem, since we will need to explicitly take into account the formation of different layers of particles as we build the triangular structure (this was implicitly taken care by the spiral formation in the HEX algorithm). The TRI snake construction will start from the seed particle $p$, which will occupy one of the triangle corners. The seed will mark two of its adjacent edges as the direction along which two of the borders of the triangle will be formed, 
 by setting $p.border[left]$ and $p.border[right]$ flags on the corresponding edges (we arbitrarily pick the edges will labels $0$ and $1$ out of $p$ in our algorithm). These directions will be propagated by the particles that end up on one of the two sides. 
The seed starts the snake formation by setting the flag $p.snakedir$ on its 0-labeled edge.
From there on, Algorithm~\ref{alg:retiredConditionTriangle} will build the triangle snake layer by layer,
alternating 
going "to the left" and "to the right". Every time the snake touches one of the borders (Case 2 of Algorithm~\ref{alg:retiredConditionTriangle}), it  sets up the rules for starting a new layer by setting the snake direction flags accordingly, first on the last particle of a layer (the one that just touched the border, in Case 2) and then on the first particle of the newly formed layer (Case 3). If a new layer is not needed, the snake proceeds to fill additional positions on the current layer (Case 1).
 Figure~\ref{fig:trianglesnapshots} illustrates this approach through some snapshots of the execution of the TRI algorithm. The proof of the following theorem appears in the Appendix.  

\begin{algorithm}[htb]
	\begin{algorithmic}
	\If {$p$ is a contracted root}
  \If {$p$ has an adjacent edge $i$ to $p'$ with a flag $p'.snakedir$, \mbox{\ \ \  \ \ \  \ \ \ \ }where $p'$ is retired} \Comment{\textbf{retired condition}}
	\State $bordertype=$ {\sc Border} $(p)$
	\If {$bordertype =$ null } 
	\State {\small \Comment{\textbf{\textbf{Case 1}: continue on the same layer\mbox{\ \ \  \  \ \ \ \ }}}}
		\State $p$ sets $p.snakedir$ for edge opposite to $i$  \mbox{\ \ \  \ \ \  \ \ \ \ \ \ }(i.e., edge $(i+3) \mod 6$)  
	\Else
			\State Let $q$ be the border particle connected to $p$
			\State Let $j$ be the edge of $p$ opposite to the edge \mbox{\ \ \  \ \ \  \ \ \ \ \ \ } connecting $p$ to $q$
			\State $p$ sets $p.border[bordertype]$ on edge $j$ 
			\If{$p' \not= q$} 
			\State {\small \Comment{\textbf{\textbf{Case 2}: start a new layer\mbox{\ \ \  \ \ \  \ \ \ \ \  \ \ \ \ \  \ \ }}}}
						\State $p$ sets $p.snakedir$  for  edge $j$

                           \Else 
													\State {\small \Comment{\textbf{\textbf{Case 3}: snake direction from border\mbox{\ \ \ }}}}
                              \If {$bordertype= left$} 
								\State $p$ sets $p.snakedir$  for  edge $(i+5)\mod 6$  
						\Else  
								\State $p$ sets $p.snakedir$  for edge  $(i+1)\mod 6$ 
						\EndIf
												  \EndIf
		\EndIf
		\State $p$ becomes retired
		\EndIf
	\EndIf
\end{algorithmic}

\vspace{.05in}
		
	{\sc Border} $(p)$: \\
    \vspace{-.2in}
   \begin{algorithmic}
				\If {$p$ has an adjacent edge $k$ to a particle $q$ with a flag $q.border[bordertype]$, where $bordertype \in \{left, right\}$
				}
					\State \textbf{return} $bordertype$
				\Else
					\State \textbf{return} null
				\EndIf
				
		\end{algorithmic}
	\caption{Snake Formation for TRI}
	\label{alg:retiredConditionTriangle}
\end{algorithm}

\begin{figure*}
  \includegraphics[width = 0.24\textwidth]{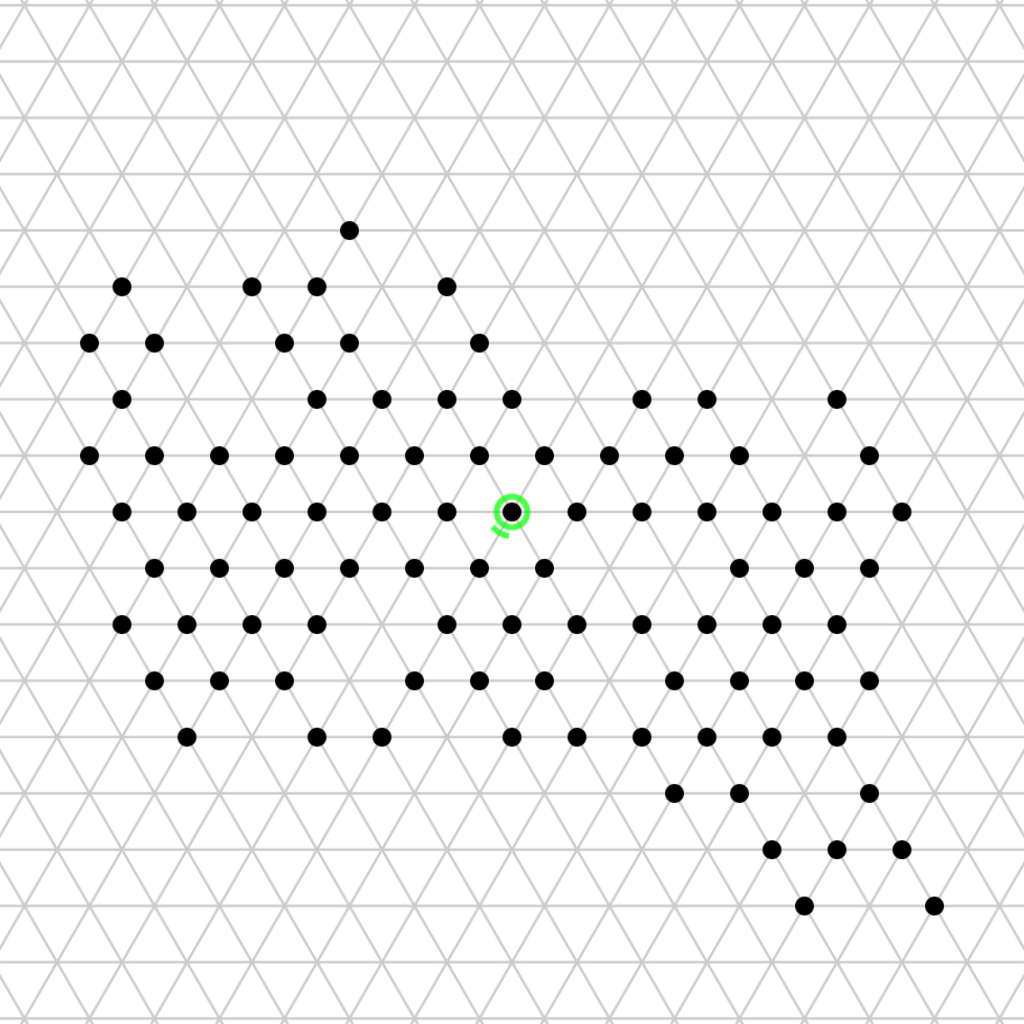}
	\includegraphics[width = 0.24\textwidth]{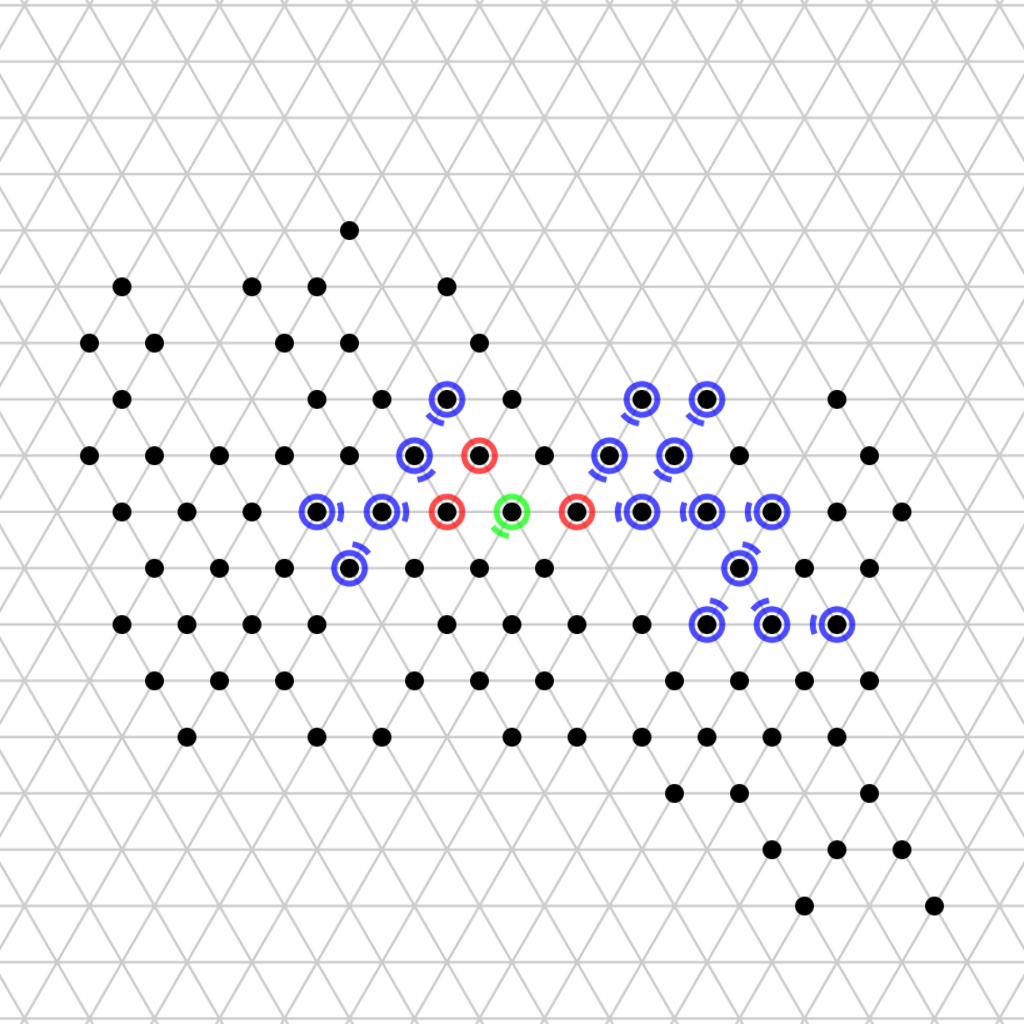}
  \includegraphics[width = 0.24\textwidth]{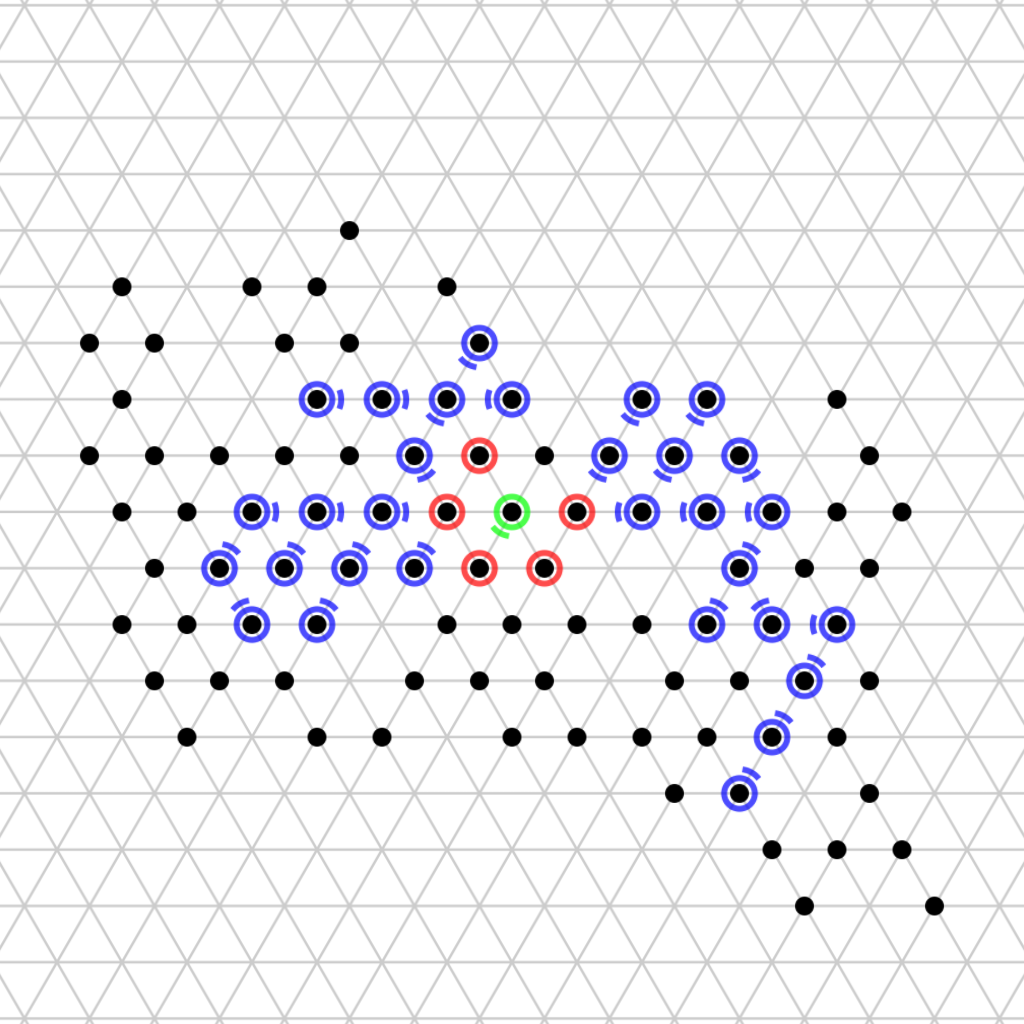}
	\includegraphics[width = 0.24\textwidth]{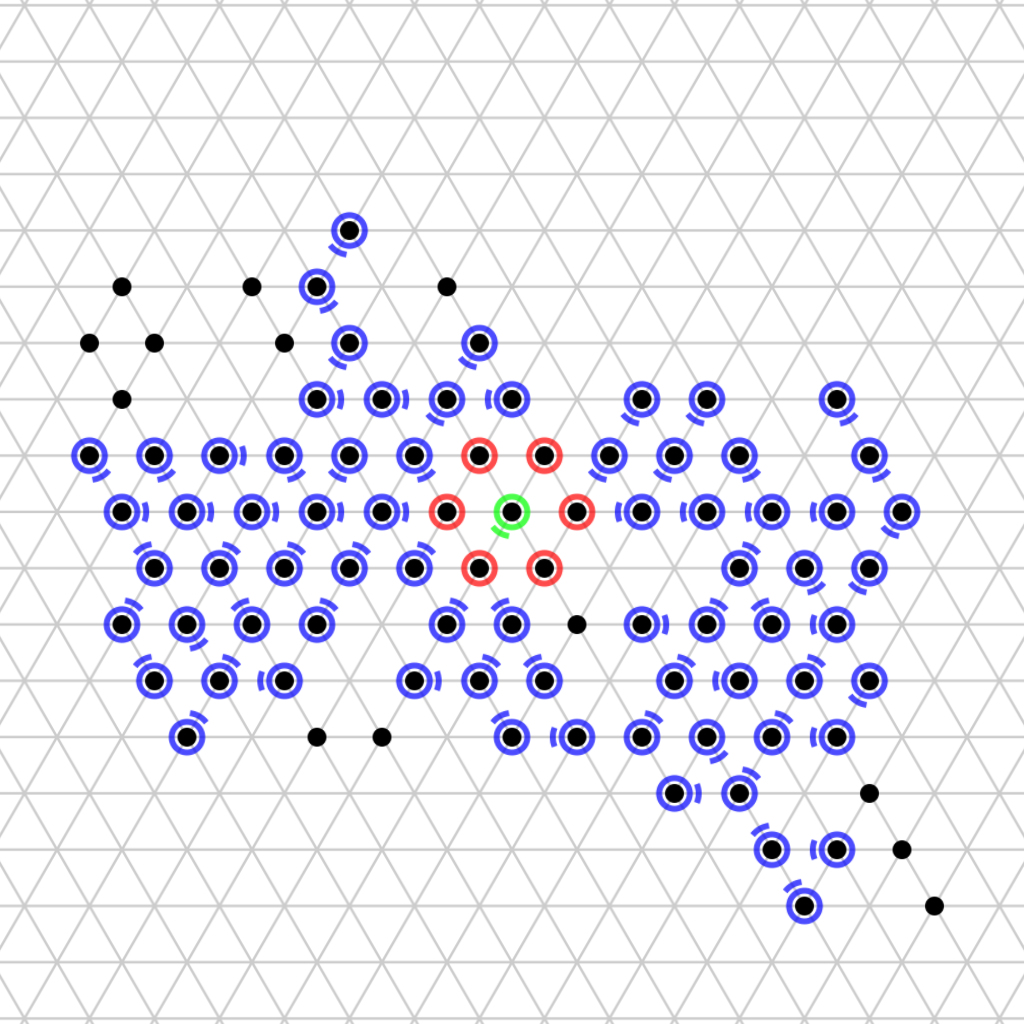}
	\\
	\\
	\includegraphics[width = 0.24\textwidth]{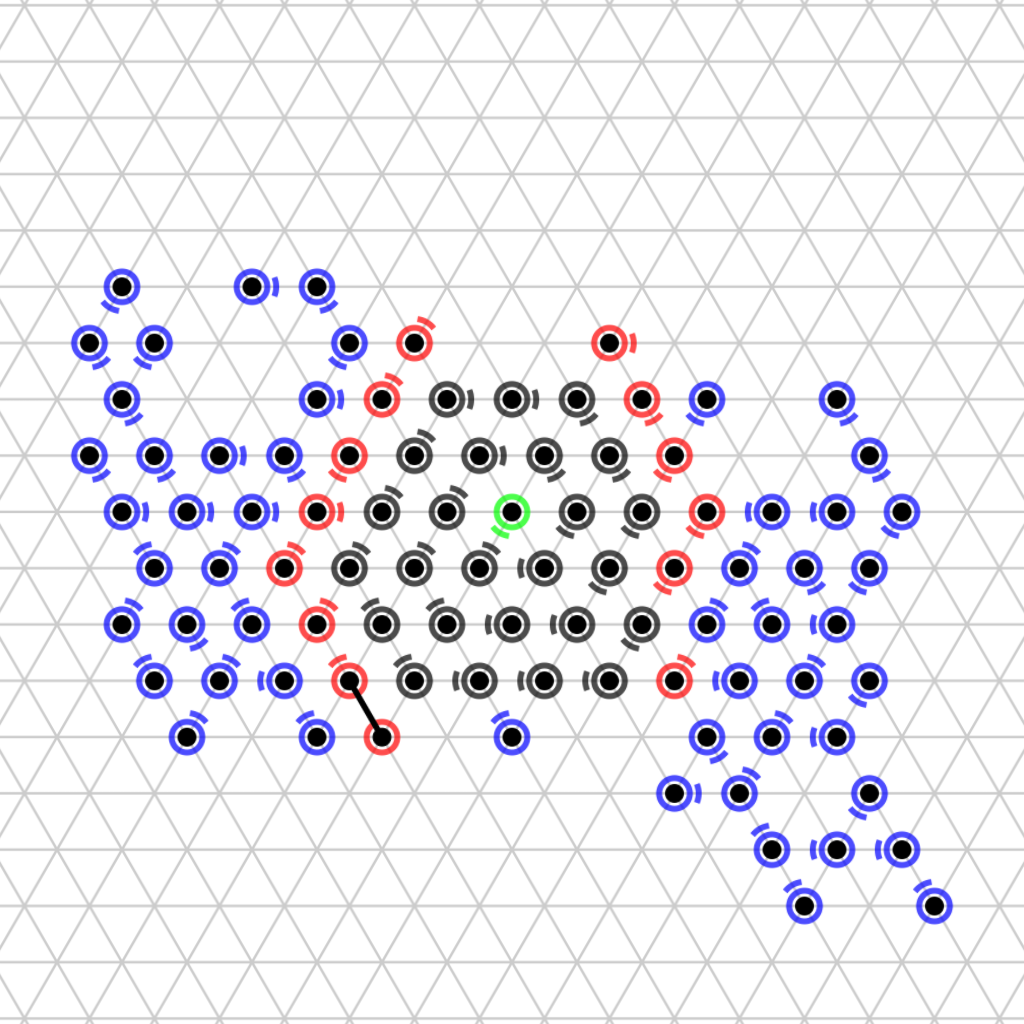}
  \includegraphics[width = 0.24\textwidth]{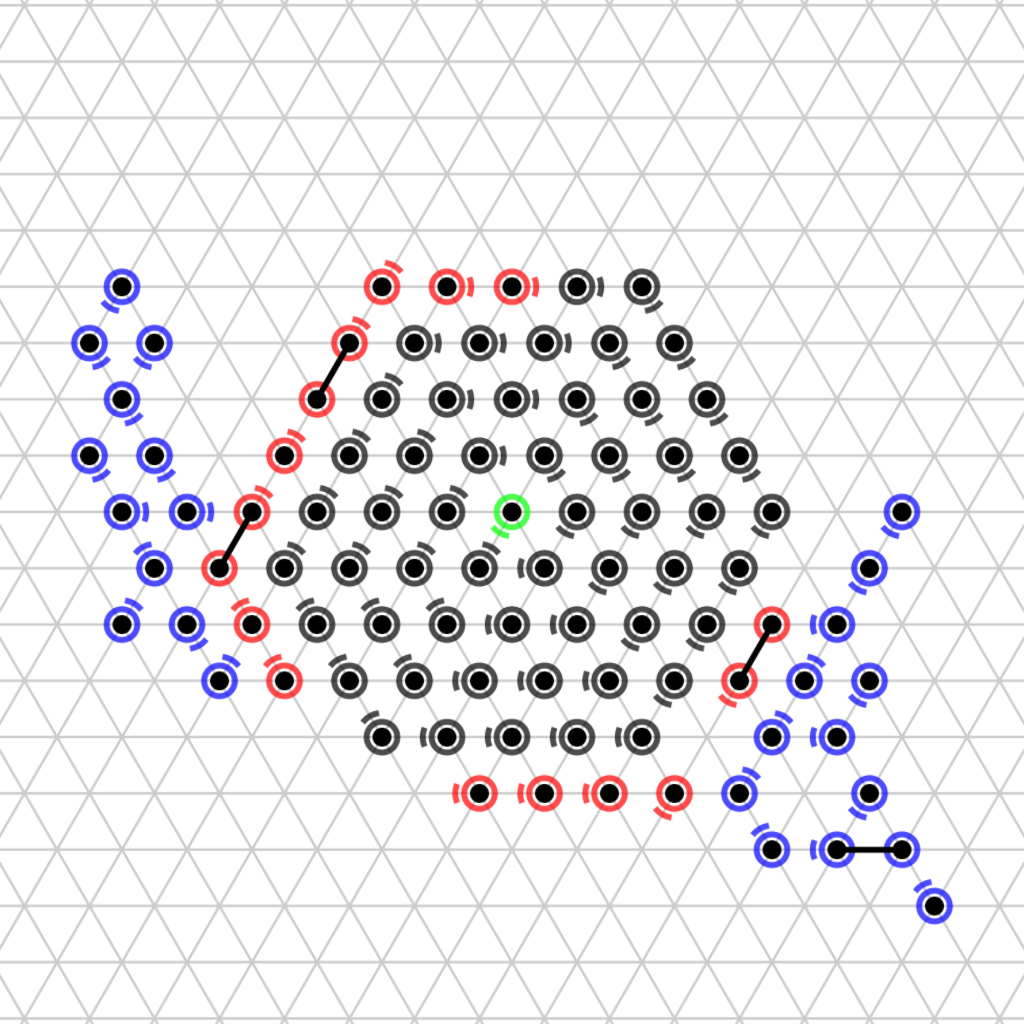}
  \includegraphics[width = 0.24\textwidth]{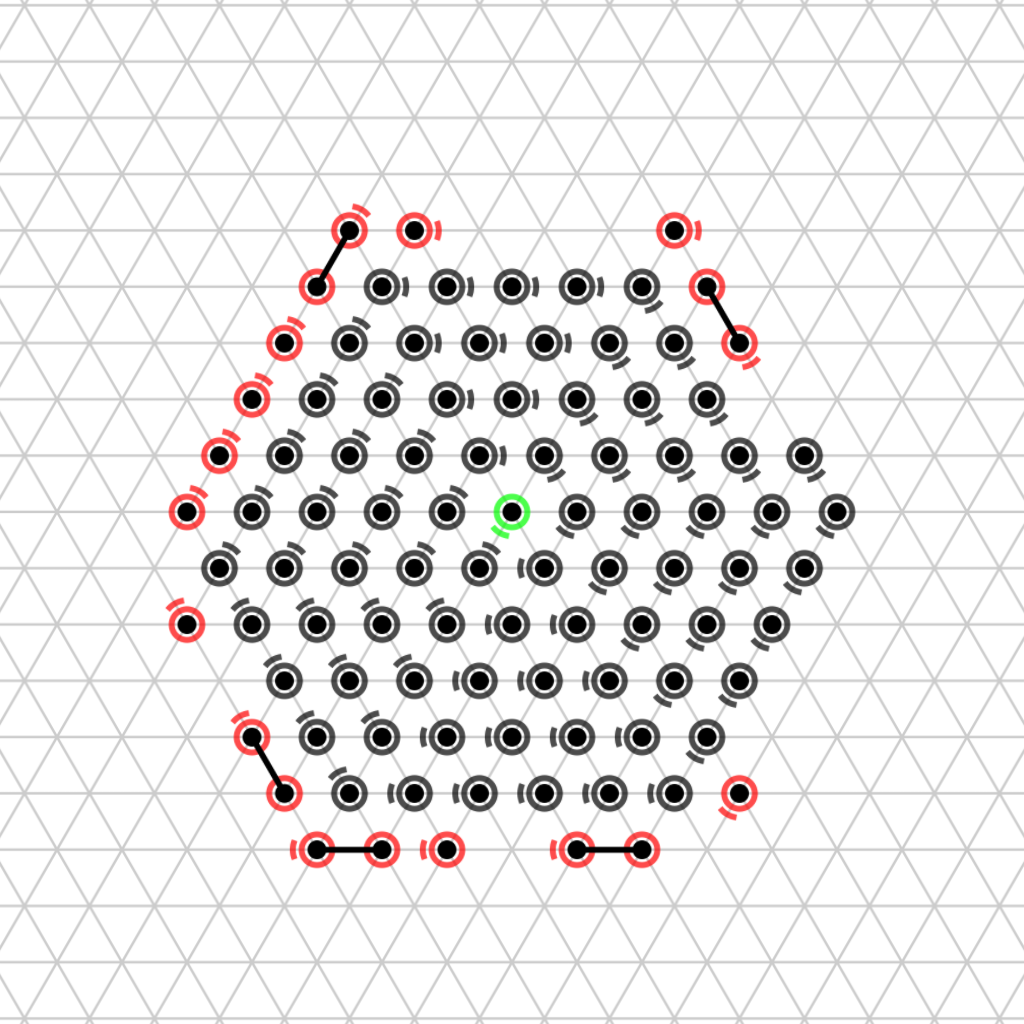}
  \includegraphics[width = 0.24\textwidth]{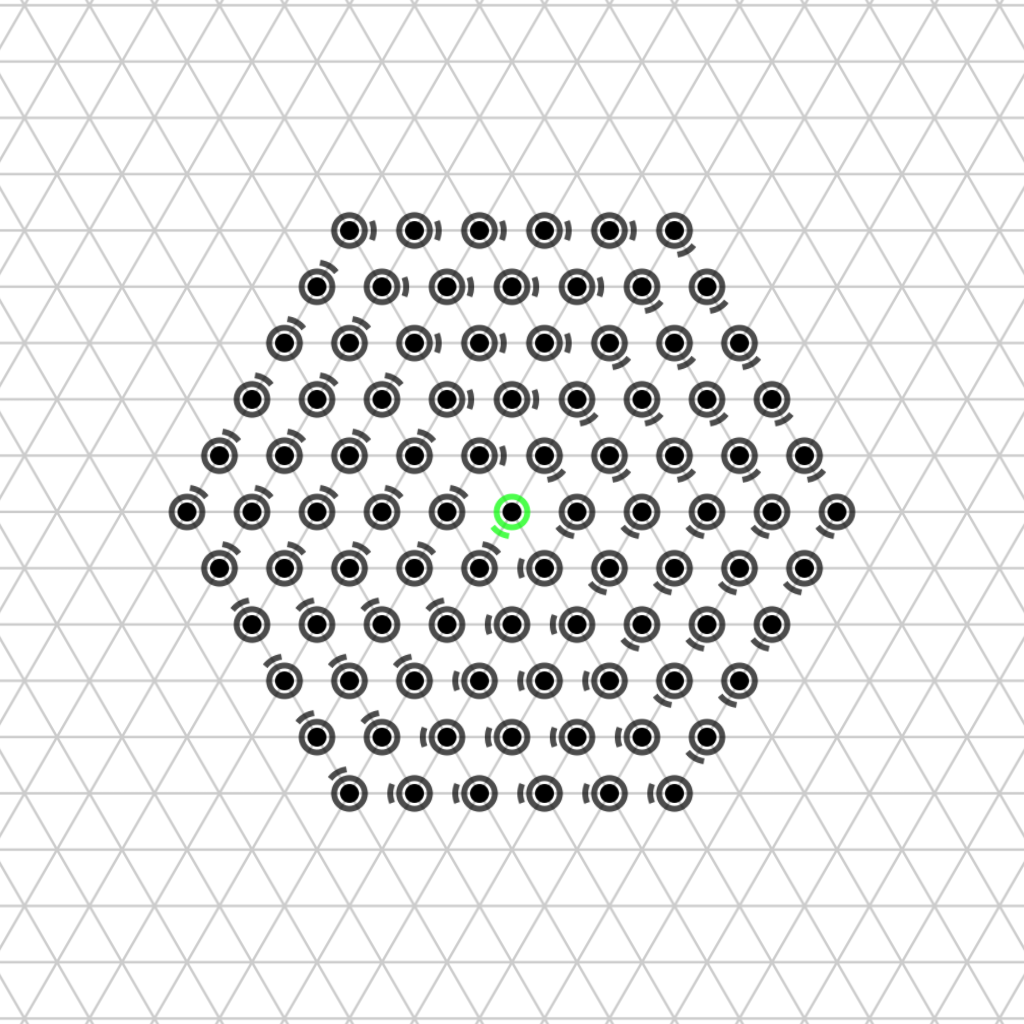}
  \caption{\small 
	Snapshots of the HEX algorithm. The seed is green, retired particles are black, roots are red and followers are blue. For a full simulation run of the algorithm see http://sops.cs.upb.de.
  }
  \label{fig:hexagonsnapshots}
\end{figure*}

\begin{figure*}
  \includegraphics[width = 0.24\textwidth]{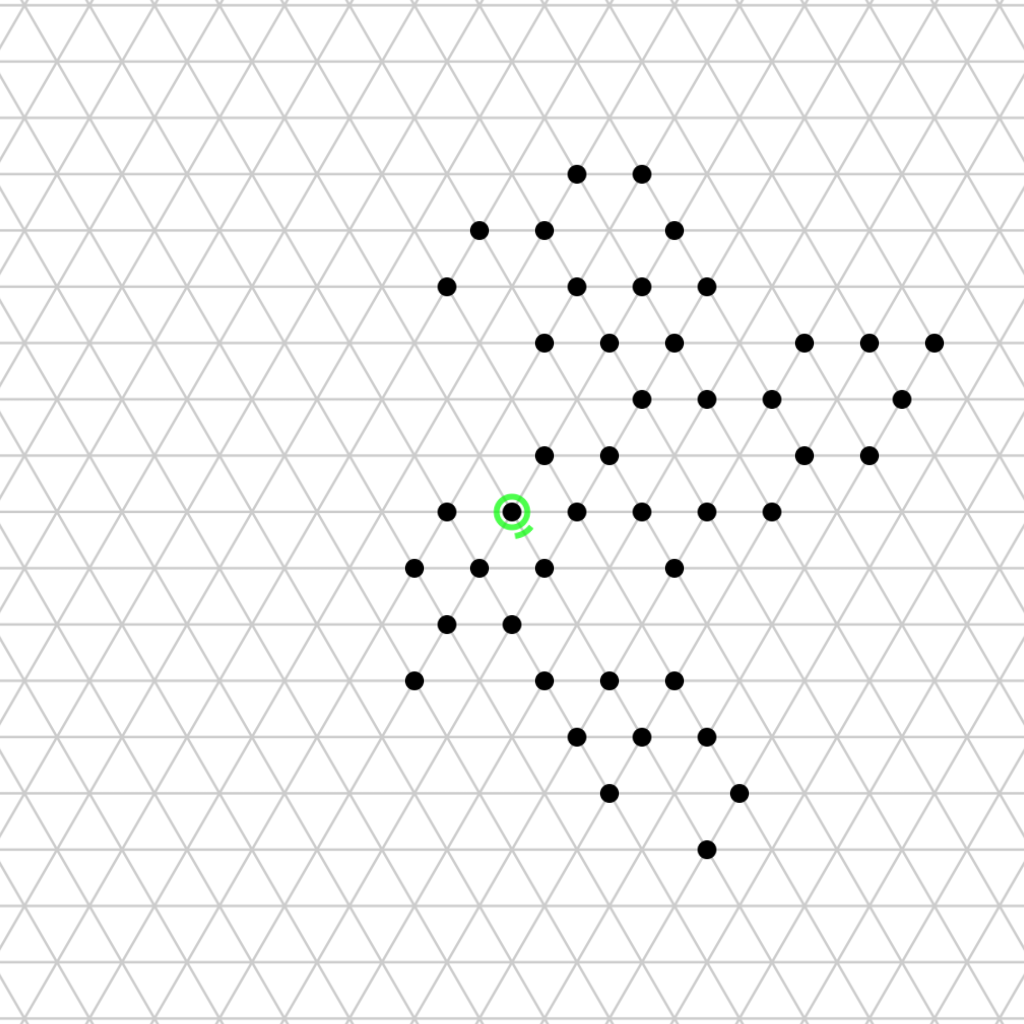}
	\includegraphics[width = 0.24\textwidth]{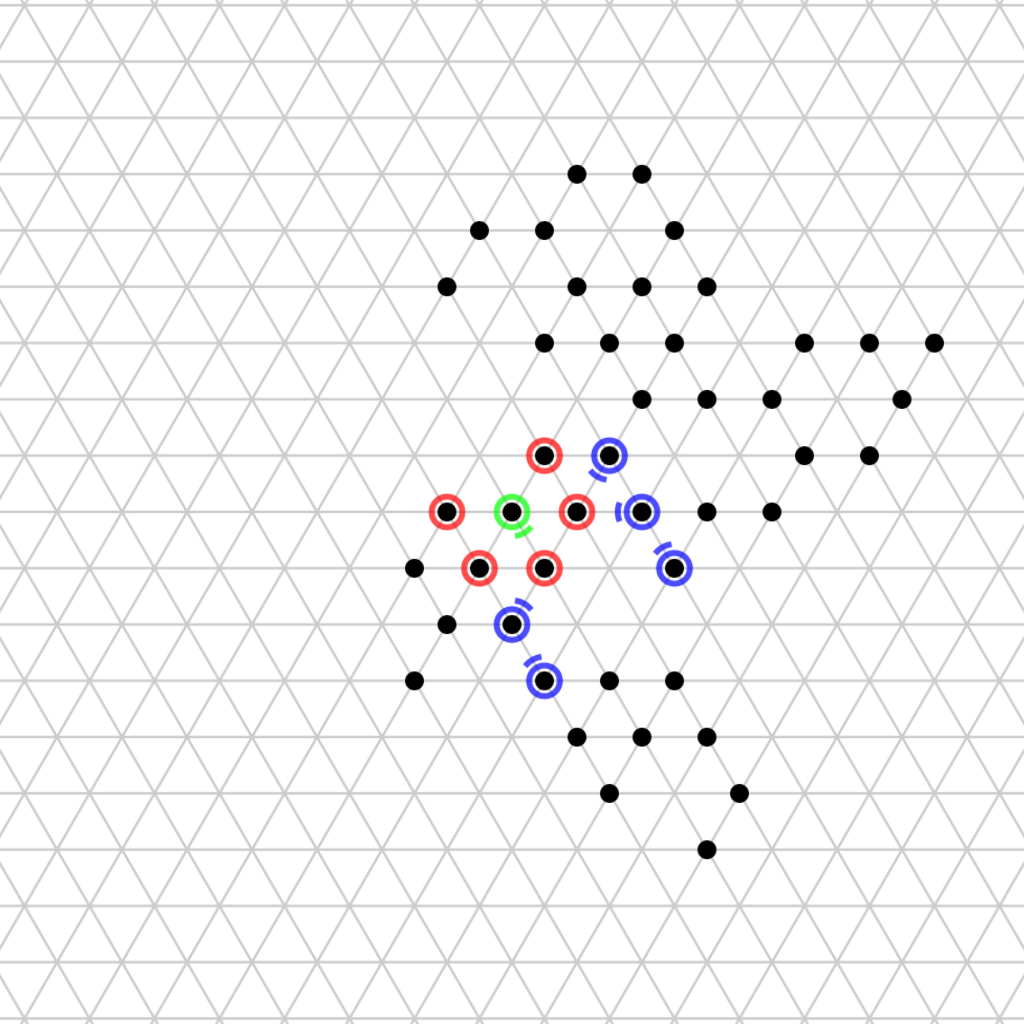}
	\includegraphics[width = 0.24\textwidth]{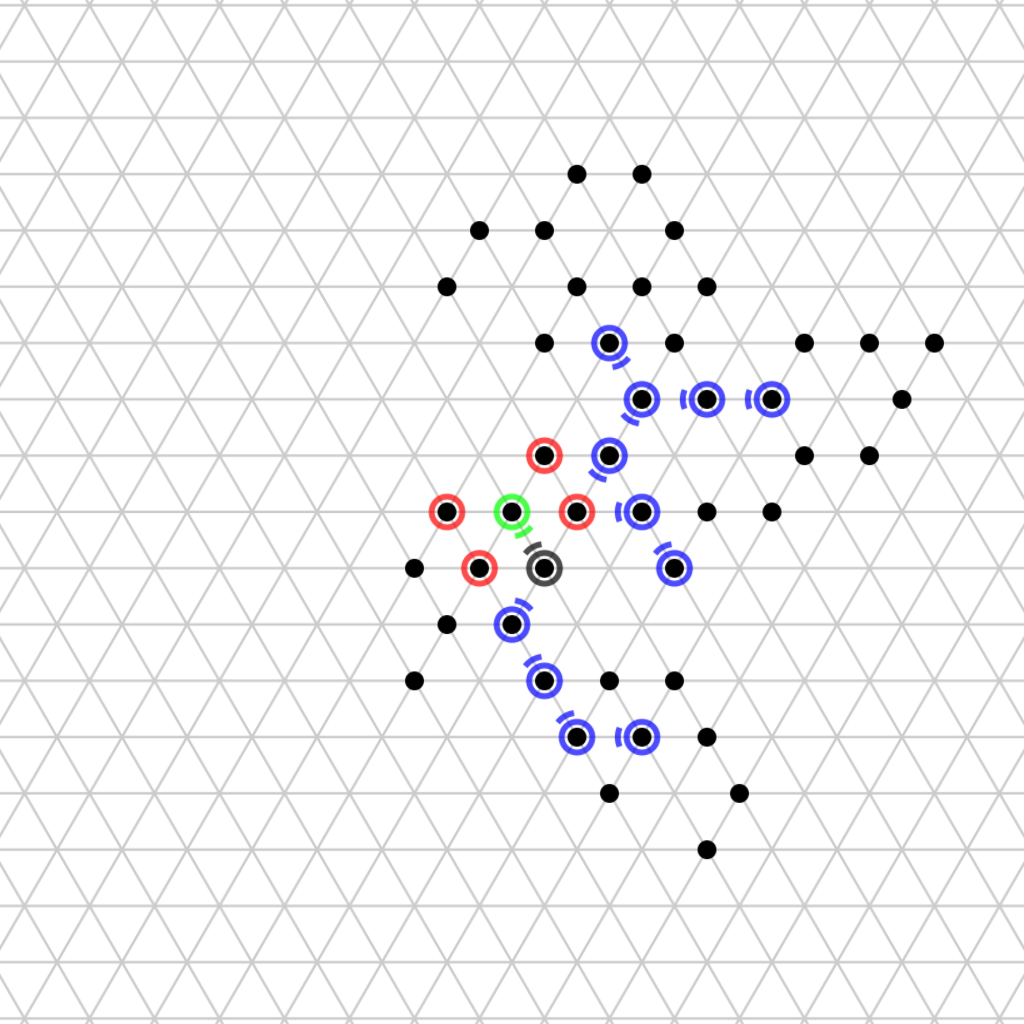}
  \includegraphics[width = 0.24\textwidth]{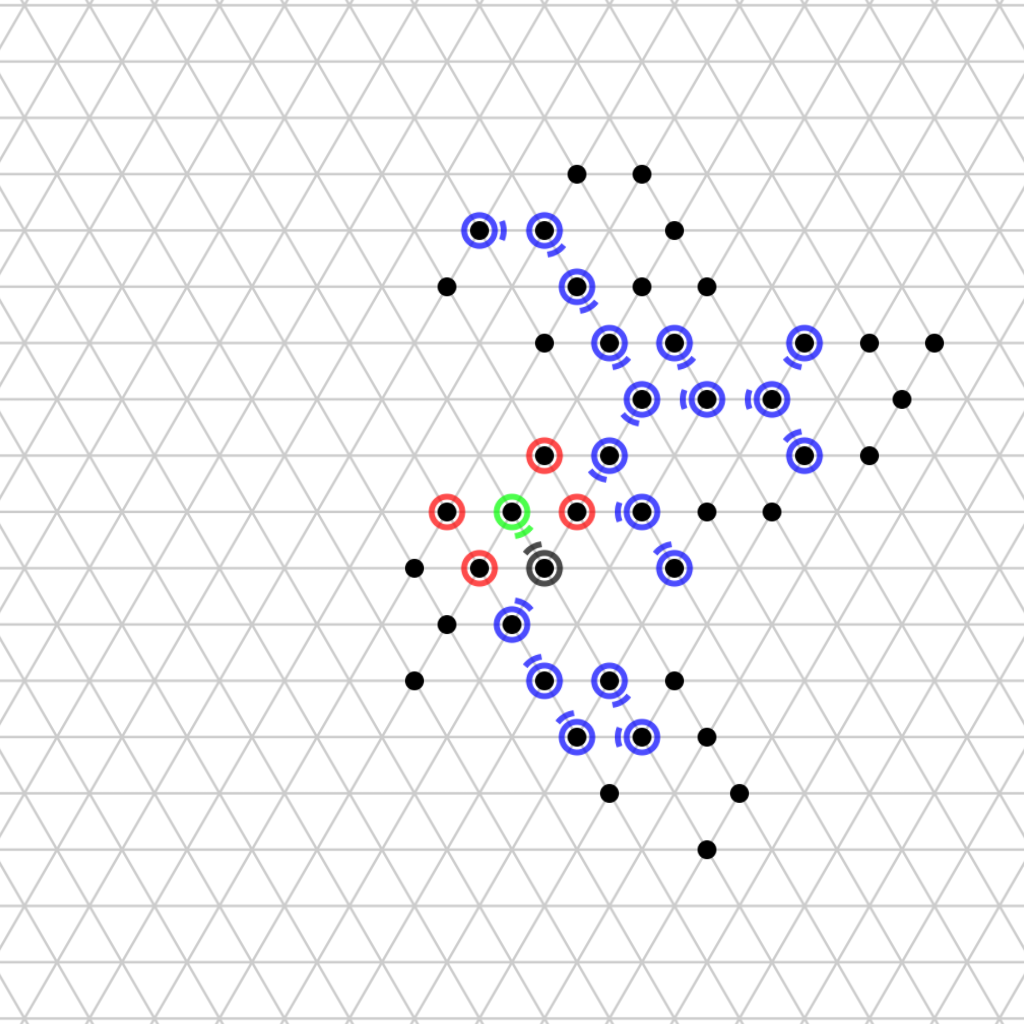}
  \\
	\\
	\includegraphics[width = 0.24\textwidth]{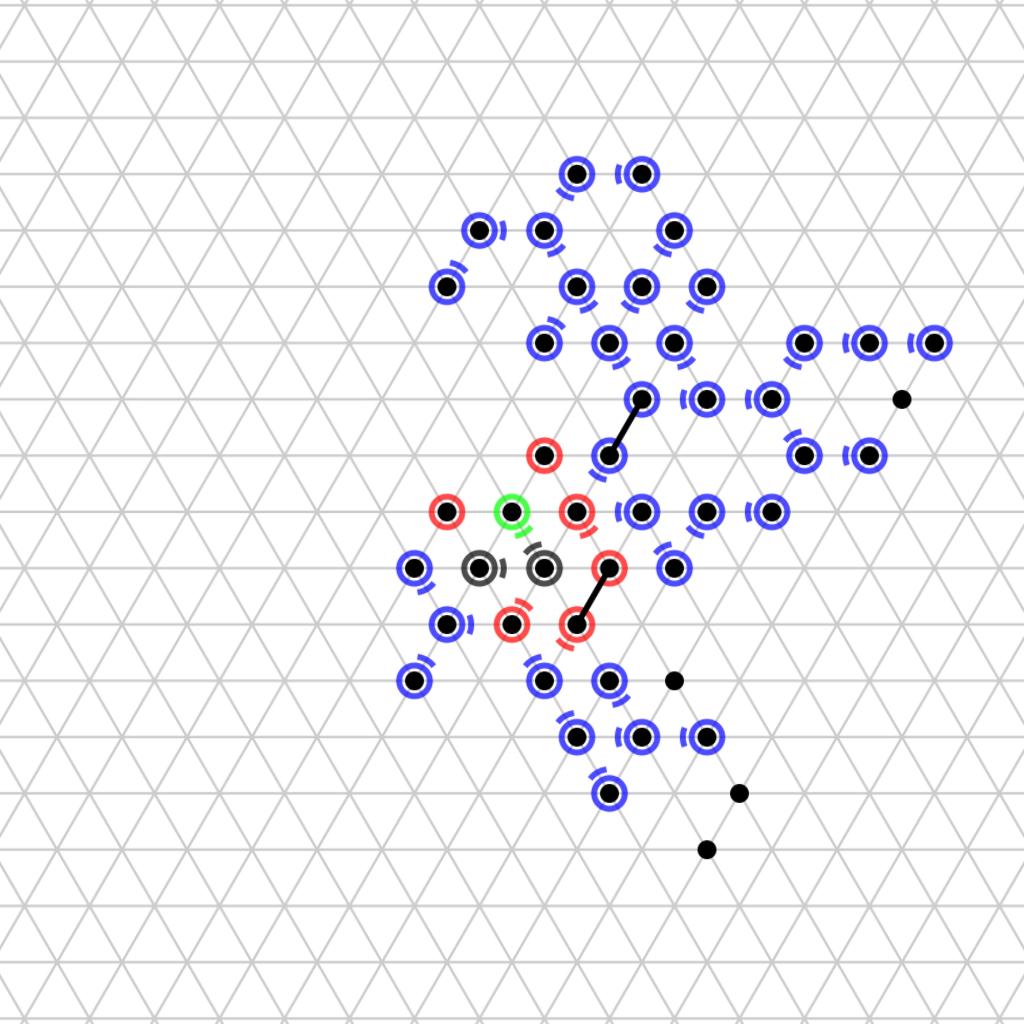}
  \includegraphics[width = 0.24\textwidth]{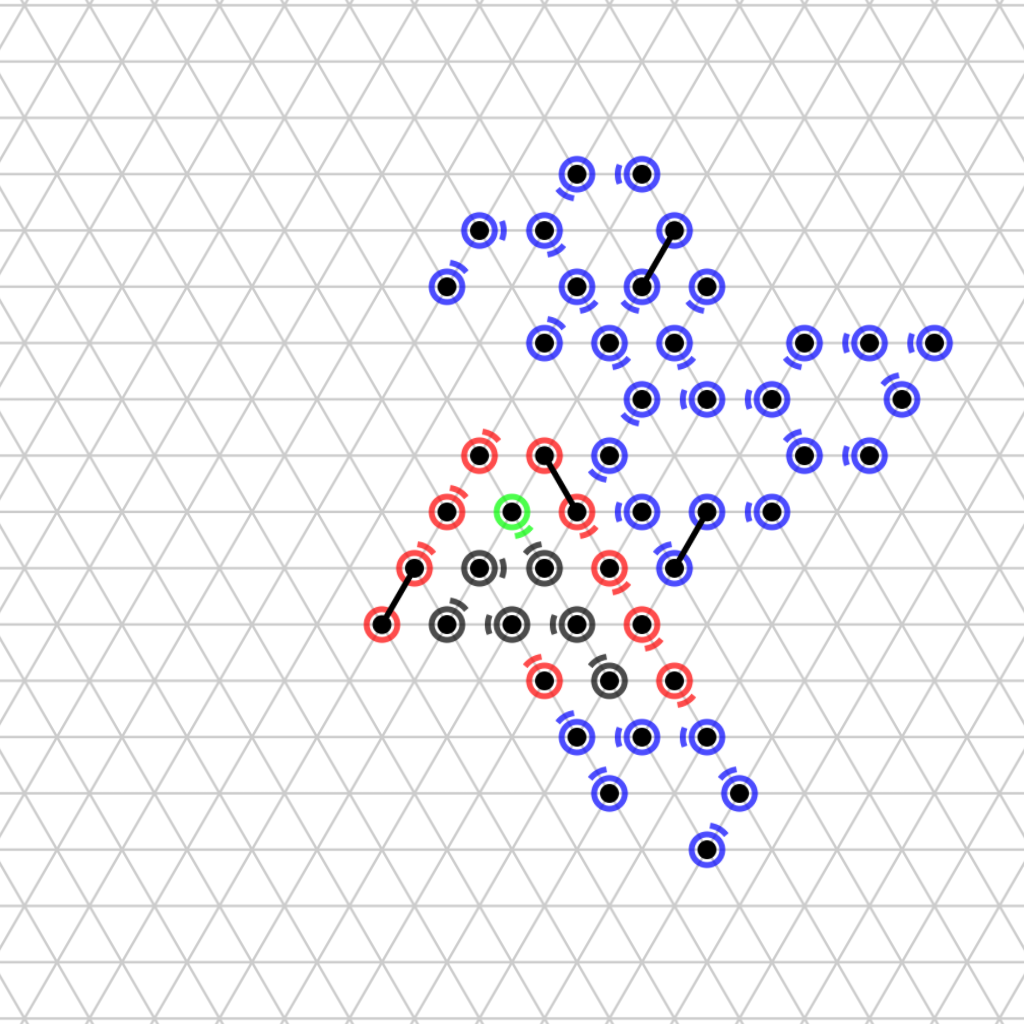}
  \includegraphics[width = 0.24\textwidth]{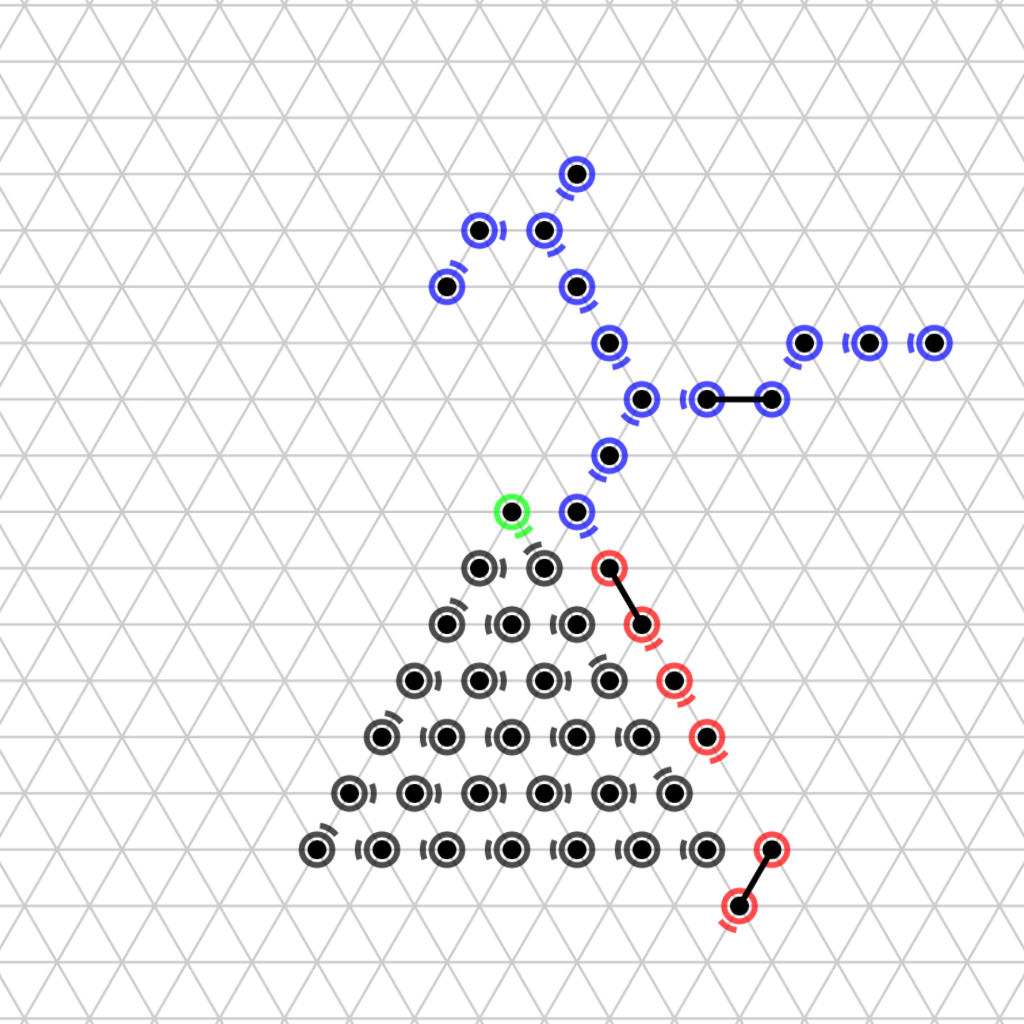}
  \includegraphics[width = 0.24\textwidth]{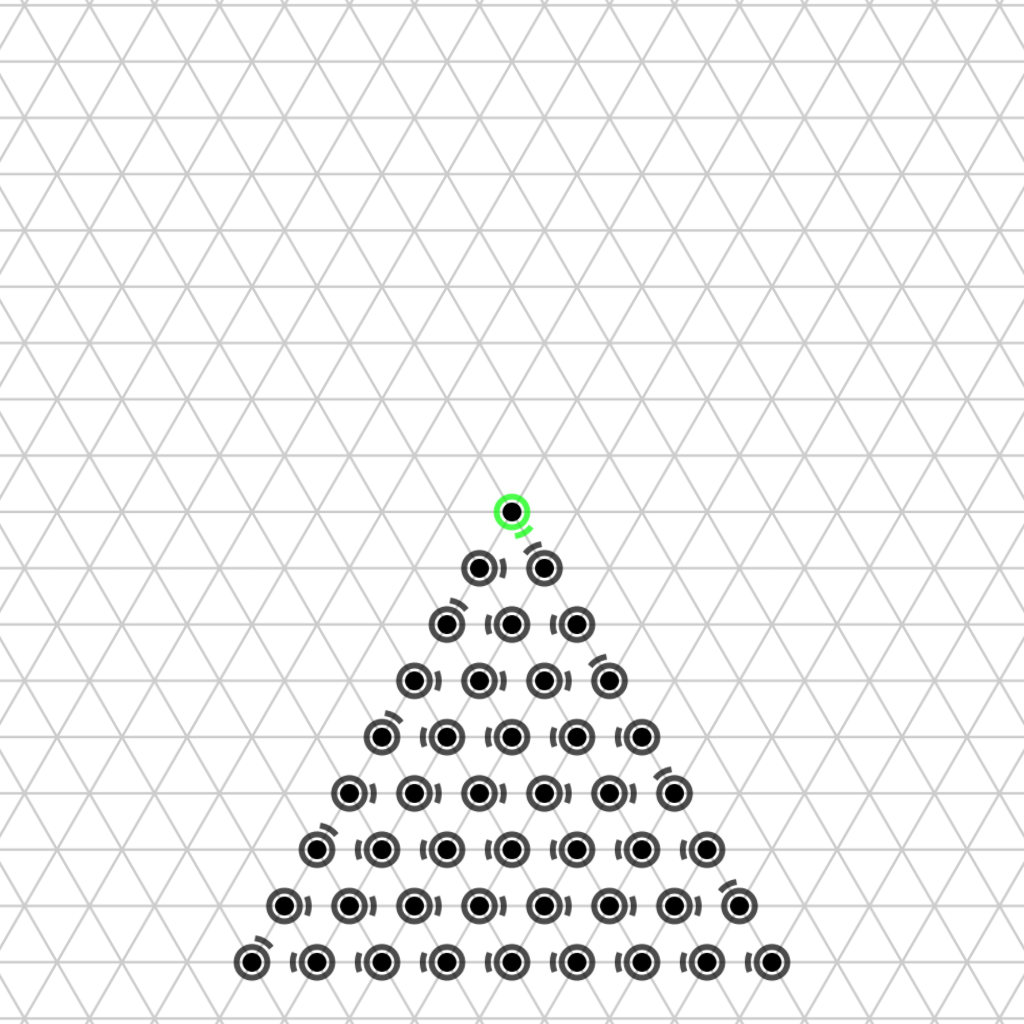}
  \caption{\small
	 Snapshots of the TRI algorithm. The seed is green, retired particles are black, roots are red and followers are blue. For a full simulation run of the algorithm see http://sops.cs.upb.de.
  }
  \label{fig:trianglesnapshots}
\end{figure*}

	\vspace{-.1in}
\begin{theorem}
	\label{thm:triangle}
	Our algorithm correctly solves the TRI problem in worst-case optimal $O(n^2)$ work.
\end{theorem}

	\vspace{-.1in}
\section{Conclusion}
We presented a general algorithmic framework for shape formation problems in SOPS  that combines our spanning forest algorithmic primitive with a snake formation primitive. 
We have shown that by carefully determining how to grow the appropriate snake structure, we were able to solve the HEX and TRI problems.
We can easily to extend our snake primitive to build other shapes, such as square or rectangular shapes, or diamonds. It would be interesting to characterize all the general shapes that could be solved with our approach on $\Geqt$ (and possibly also for other infinite regular grid graphs, namely the square grid graph and the hexagonal grid graph, if we considered those as the underlying graph $G$ in the geometric amoebot model). Finally, we would like to evaluate the performance of our algorithms in terms of the worst-case number of asynchronous rounds necessary for termination.

{\small\bibliographystyle{abbrv}
\bibliography{literature}}  

\newpage

\appendix
\section{Analysis}
\subsection{Analysis of the Spanning Forest Algorithm }
\label{sec:analysisSpanningTree}

We say followers and roots 
 particles are \emph{active}.
As specified in Algorithm~\ref{alg:spanningForestAlgorithm}, only followers can set the flag $p.parent$.

The first three lemmas demonstrate some properties that hold during the execution of the spanning forest procedure
 and will be used later to analyze the proposed algorithms for HEX and TRI problems. 

\begin{lemma}
    \label{lem:successor}
	For a follower $p$, the node  indicated by $p.parent$ is occupied by an active particle.
\end{lemma}
\begin{proof}
    Consider a follower $p$ in any configuration during the execution of Algorithm~\ref{alg:spanningForestAlgorithm}.
    Note that $p$ can only become a  follower from an inactive state,
    and once it leaves the follower state it will not switch to that state again.
    Consider the first configuration $c_1$ in which $p$ is a follower.
    In the configuration $c_0$ immediately before $c_1$, $p$ must be inactive
    and it becomes a follower because of an active particle $p'$ occupying the position indicated by $p.parent$ in $c_0$.
    The particle $p'$ is still adjacent to the edge flagged by $p.parent$ in $c_1$.
    Now assume that $p.parent$ points to an active particle $p'$ in a configuration $c_i$,
    and that $p$ is still a follower in the next configuration $c_{i+1}$ that results from executing an action $a$.
    If $a$ affects $p$ and $p'$, the action must be a handover in which $p$ updates its flag $p.parent$
    such that $p.parent$ may be moved to the edge that now connects $p$ to $p'$ in $c_{i+1}$.
    If $a$ affects $p$ but not $p'$,
    it must be a contraction in which $p.parent$ does not change and still points to $p'$.
	If $a$ affects $p'$ but not $p$, there are multiple possibilities.
 	The particle $p'$ might switch from follower to root state,
 	 or from root to retired state,
	 	or it might expand,
 	none of which violates the lemma.
 	Furthermore, $p'$ might contract.
 	If $p.parent$ points to the head of $p'$, $p'$ is still adjacent to the edge flagged by $p.parent$ in $c_{i+1}$.
 	Otherwise, $p$ is a child adjacent to the tail of $p'$ in $c_i$ and therefore the contraction must be part of a handover. 
	As $p$ is not involved in the action, the handover must be between $p'$ and a third active particle $p''$.
 	It is easy to see that after such a handover $p.parent$ points to either $p'$ or $p''$.
	Finally, if $a$ affects neither $p$ nor $p'$, $p.parent$ will still point to $p'$ in $c_{i+1}$.
\end{proof}

Based on Lemma~\ref{lem:successor}, 
we define a directed graph $A(c)$ for a configuration $c$ as follows.
$A(c)$ contains the same nodes as the nodes occupied in $\Geqt$ by the set of particles in $c$.
For every expanded particle $p$ in $c$, $A(c)$ contains a directed edge from the tail to the head of $p$,
and for every follower $p'$ in $c$, $A(c)$ contains a directed edge from the head of $p'$ to $p'.parent$.

\begin{lemma}
	\label{lem:forest}
	The graph $A(c)$ is a forest, and if there is at least one active particle,
    every connected component of inactive particles contains a particle that is connected to an active particle.
\end{lemma}
\begin{proof}
    In an initial configuration $c_0$, all particles are inactive and therefore the lemma holds trivially.
    Now assume that the lemma holds for a configuration $c_i$.
    We will show that it also holds for the next configuration $c_{i+1}$ that results from executing an action $a$.
    If $a$ affects an inactive particle $p$, this particle either becomes a follower or a root.
	In the former case $p$ joins an existing tree, and in the latter case $p$ forms a new tree in $A(c_{i+1})$.
    In either case, $A(c_{i+1})$ is a forest and the connected component of inactive particles
    that $p$ belongs to in $c_i$ is either non-existent or connected to $p$ in $c_{i+1}$.
	If $a$ affects only a single particle $p$ that is in state follower, this particle can contract or become a
	 root.
	In the former case, $p$ has no child 
	$p'$ such that $p'.parent$ 
	is the tail of $p$
	and also $p$ has no inactive neighbors.
	Therefore, the contraction of $p$ does not disconnect any follower or inactive particle
	and, accordingly, does not violate the lemma.
	In the latter case, $p$ becomes a root of a tree which also does not violates the lemma.
    If $a$ involves only a single particle $p$ that is in state root, $p$ can expand or contract 
		  or become retired.
    An expansion 
		 and becoming retired 
		 trivially cannot violate the lemma and the argument for the contraction is the same as for the contraction of a follower above. 
    Finally, if $a$ involves two active particles in $c_i$, these particles perform a handover.
    While such a handover can change the parent relation among the nodes, it cannot violate the lemma.
\end{proof}

The following lemma shows that the spanning forest always makes progress, by showing that as long as the roots keep moving, the remaining particles will eventually follow.

\begin{lemma}
	\label{lem:contract}
	An expanded particle eventually contracts.
\end{lemma}
\begin{proof}
	Consider an expanded particle $p$ in a configuration $c$.
	Note that $p$ must be active.
	If there is an enabled action that includes the contraction of $p$,
	that action will remain enabled until $p$ eventually contracts when $p$ is validated in the current round.
	  So assume that there is no enabled action that includes the contraction of $p$.
	According to Lemma~\ref{lem:forest} and the transition rule from inactive to active particles,
	at some point in time all particles in the system will be active.
	If the contraction of $p$ becomes part of an enabled action before this happens, $p$ will eventually contract.
	So assume that all particles are active but still $p$ cannot contract.
	If $p$ has no children, the isolated contraction of $p$ is an enabled action which contradicts our assumption.
	Therefore, $p$ must have children.
	
	Furthermore, $p$ must read at least one child $p'$ having its $p'.parent$ flag pointing towards $p$ over its tail and all children having their parent flags pointing towards $p$'s tail must be expanded
	as otherwise $p$ could again contract as part of a handover. 
	 	If $p'$ would contract, a handover between $p'$ and $p$ would become an enabled action.
	We can apply the complete argument presented in this proof so far to $p'$
	and so on backwards along a branch in a tree in $A(c)$ until we reach a particle that can contract.
	We will reach such a particle by Lemma~\ref{lem:forest}.
	Therefore, we found a sequence of expanded particles that starts with $p'$
	and ends with a particle that eventually contracts.
	The contraction of that last particle will allow the particle before it in the sequence to contract and so on.
	Finally, the contraction of $p$ will become part of an enabled action and therefore $p$ will eventually contract.
\end{proof}

\subsection{Hexagonal Shape Formation Analysis}
\label{sec:analysisHexagon}

Here, we show that the algorithmic primitives proposed in Section~\ref{sec:HexagonFormation} solve the HEX problem correctly.

\begin{theorem}
\label{thm:solveHexagon}
Our algorithm solves the HEX problem.
\end{theorem}
\begin{proof}
	We need to show that the algorithm terminates and that when it does, the system is in the shape of a hexagon.
	 According to Lemma~\ref{lem:forest}, every particle $p$ eventually activates.
    According to the spanning forest algorithm, if $p$ is adjacent to the retired structure (initially the structure only contains the seed particle), it becomes a root and moves in a clockwise manner around the retired structure until it eventually reaches the valid position that can extend the hexagon and becomes retired. By contradiction, assume $p$ never becomes retired. Since the number of particles is bounded (and therefore the size of the formed retired structure is bounded), there must be an infinite number of configurations $c_i$ where $p$ had a root particle blocking its desired  clockwise movement around the hexagona retired structure. Let $p'$ be the last root $p$ sees as its clockwise neighbor over the retired structure (since once a particle becomes a root, it will stay connected to the hexagonal retired structure and always attempt to move in a clockwise manner, $p'$ is well-defined).
     Applying the same argument inductively to $p'$, we will get an infinite sequence of roots on the retired structure that never touch a valid spot pointed by $q.snakedir$ flag of an already retired particle $q$,  a contradiction, since the current retired structure (and the number of retired particles) is bounded. Therefore, every root eventually changes into a retired state.
	From Algorithm~\ref{alg:spanningForestAlgorithm}, every follower in the
 neighborhood of a retired particle becomes a root. For every root $q$ with at least one follower child, let $c$ be the first
configuration when $q$ becomes retired.
If $q$ still has any child
 in $c$ then all of its children $p$ become roots. Applying this argument recursively we will reach to a configuration such that there exists no root $q$ having a follower child which proves that eventually every follower becomes a root.
 Putting it all together, eventually all particles become retired and the algorithm terminates.

 Note that it also follows from the argument above that the set of retired particles at the end of the algorithm forms a connected structure (since the particles start from a connected configuration and never get disconnected through the process). 

Now, we need to prove the correctness, i.e., that the resulting structure of retired particles is in a hexagonal form.
 Initially the hexagon only contains the seed particle, therefore the claim holds trivially.
 By induction, let's assume $c$ is the first configuration in which the current formed structure of the retired particles contains $k$ retired particles and by induction hypothesis, assume that those particles form a valid hexagonal shape using $k$ particles. 
 According to  Algorithm~\ref{alg:retiredConditionHexagon}, the only way a  root $p$ can become the $(k+1)^{\text{\tiny th}}$ retired particle during or after $c$, is if it occupies the next valid position pointed by the flag 
  $q.snakedir$, where 
  $q$ was the $k$-th particle to join the hexagonal shape. According to induction hypothesis, the $k$ first retired particles form 
  a hexagonal shape. By pointing to the  next adjacent position in counter-clockwise direction around the outermost retired particles in the current hexagonal structure, the flag $q.snakedir$ points to the next position (according to counter-clockwise direction)
  on the last formed layer of the retired structure, or to a starting position on the next layer once the current layer is full,
 proving the correctness of the constructed shape.
\end{proof}

We would like to measure the amount of the work of the proposed algorithm.

\begin{lemma}
	\label{lem:worstCaseHexagon}
	The worst-case work required by any algorithm to solve the HEX problem is $\Omega(n^2)$.
\end{lemma}
\begin{proof}
Consider a line of $n$ particles on $\Geqt$, where the seed particle is located on one end of the line, 
 as an initial configuration of the particles. We label the particles connected to the seed starting with number 0 for the particle adjacent to the seed.  
 The particle labeled $i > 1$ requires at least $2(i-1-\left\lceil (i-1)/M_{i-1}\right\rceil) \geq 2(i-1-\left\lceil (i-1)/6\right\rceil) $ movements until it can lie contracted on the retired structure  
 where $M_j$, $M_j \geq 6$ and $j \geq 1$, indicates the capacity (i.e., the number of the retired particles) of the layer that the retired particle with label $j$ belongs to. 
	 Therefore, any algorithm requires at least  
	 $2\sum_{i = 2}^{n-1} (i-1-\left\lceil (i-1)/6\right\rceil) 
	= \Omega(n^2)$ work. 
\end{proof}
\begin{theorem}
	The algorithm proposed for HEX terminates in $O(n^2)$ work.
\label{thm:workHexagon}
\end{theorem}
\begin{proof}
To prove the upper bound, we simply show that every particle executes $O(n)$ movements. The theorem then follows.
    Consider a particle $p$.
    While $p$ is in inactive or a retired state, it does not move.
    Let $c$ be the first configuration when $p$ becomes a follower. 
     Consider the directed path in $A(c)$ from the head of $p$ to its root $p'$. There always is such a path since every follower belongs to a tree in $A(c)$ by
    Lemma~\ref{lem:forest}
    . Let $P= (a_0, a_1, \ldots, a_m)$ be that path in $A(c)$ where $a_0$ is the head of $p$ and $a_m$ is a child of  $p'$. According to Algorithm~\ref{alg:spanningForestAlgorithm}, $p$ attempts to follow $P$ by sequentially expanding into the nodes $a_0, a_1, \ldots, a_m$. The length of this path is bounded by $2n$ and, therefore, the number of movements $p$ executes while being a follower is $O(n)$.
     Once $p$ becomes a root, it only performs expansions and contractions around the retired structure 
     until it reaches one of the valid positions on the hexagon. 
 Since each root $p$ and a retired particle $q$ never connect from the same edge more than twice, 
 and since the total number of retired particles is at most $n$, therefore the number of movements is bound to $O(n)$ for $p$. 
     Therefore, the number of movements a particle $p$ totally executes is $O(n)$, which concludes the theorem. 
\end{proof}

\subsection{Analysis for Triangular Shape Formation}
\label{sec:analysisTriangle}

Now we need to show that the algorithmic primitives presented in Section~\ref{sec:TriangleFormation} solve the TRI problem correctly.

\begin{theorem}
	\label{thm:solveTriangle}
	Our algorithm solves the TRI problem.
\end{theorem}
\begin{proof}
	Again, we need to show that the algorithm terminates and that when it does, the system is the shape of a triangle. The termination part of the proof is identical to that for the HEX problem  presented in Theorem~\ref{thm:solveHexagon}, and hence it only remains to prove the correctness of the TRI algorithm. 
 Assume we have three particles as the base case (to build the smallest size perfect triangle on $\Geqt$). The seed $p*$ sets the $p*.snakedir$ flag and the $p*.border[left]$ flag on its 0-labeled edge. A root particle $q$ might have to move around the seed $p*$ until it connects to edge 0 of the seed through an edges $i$. Since $p$ sees both (border and snake) flags coming from the same particle, $p$ becomes retired while it start constructing a new layer of the triangle and sets its $p.snakedir$ flag such that the next particle continues filling this newly added layer (Case 3 of Algorithm~\ref{alg:retiredConditionTriangle}). Particle $p$ also sets $p.border[left]$ appropriately to propagate the inherited direction of the border from the seed to next layer.  
The only position that the third particle can stop on $\Geqt$ is the one pointed by $p.snakedir$ and it is trivial to see that the resulting retired structure of the three particles is in a triangular shape.  
 Let  $c$ 
be the first configuration in which the current formed structure of the retired particles contains $k$ retired particles, and let $q$ denote the $(k)^{\text{\tiny th}}$ particle to become retired. 
By induction hypothesis, assume that those $k$ particles form a triangle.
  According to  Algorithm~\ref{alg:retiredConditionTriangle}, the only way a root $p$ can become the $(k+1)^{\text{\tiny th}}$ retired particle during or after $c$, is if it occupies the valid position pointed by a flag $q.snakedir$. Depending on the location of $q$ in the triangle, three cases
 may arise. First, consider the case when $q$ is a left border particle (an analogous argument works if $q$ is a right border particle). 
 Since $q$ is the last particle  added to the current valid triangular shape,
 we either have a perfect triangle after the addition of $q$ or   
 we have a perfect triangle plus  particle $q$ as the leftmost particle on a newly created layer.
  In the former case, given the next position pointed by $q.snakedir$, the root $p$ follows Case 2 of algorithm, which means that $p$ will retire on the leftmost valid position on the next layer of the triangular structure, pointed by $q.border[left]$.
  In the latter, $p$ follows Case 3 and will  fill another position of the current layer next to $q$. In both cases the resulting retired structure still forms a valid triangular shape.  	Second, consider the situation where $q$ is not a border particle (Case 1). Therefore, $q$ is located on the last partially filled layer and $q.snakedir$ is set to point to the  next unoccupied snake spot on that layer, which is then filled by $p$, correctly extending the triangular structure, and proving the claim. 
\end{proof}

Again, we would like to measure the work of the proposed algorithm.

\begin{lemma}
	\label{lem:worstCaseTriangle}
	The worst-case work required by any algorithm to solve the TRI problem is $\Omega(n^2)$.
\end{lemma}

\begin{proof}
With a very similar argument we had in Lemma~\ref{lem:worstCaseHexagon} one can verify that it is required to have at least $2\sum_{i = 1}^{n-1} (i-1-\left\lceil (i-1)/2\right\rceil)
	= \Omega(n^2)$ work for the algorithm to terminate.
\end{proof}

\begin{theorem}
	The algorithm for TRI terminates in $O(n^2)$ work.
\label{thm:workTriangle}
\end{theorem}
\begin{proof}
Same argument we have in Lemma~\ref{thm:workHexagon} holds here too. We just need to assume a triangular shape instead of a hexagonal one. 
\end{proof}

\end{document}